\documentclass[12pt,a4paper]{article}
\usepackage{epsf, graphicx}
\usepackage{latexsym,amsfonts,amsbsy,amssymb}
\usepackage{mathrsfs}
\usepackage{amsmath,amssymb,amscd,bm,bbm,amsthm,mathrsfs}
\usepackage{amsfonts,amssymb}
\usepackage{cite}
\usepackage{indentfirst}
\usepackage{bm}
\usepackage{cite}
\usepackage{cleveref}
\usepackage{ragged2e}
\usepackage{color}
\usepackage{multirow,booktabs}
\usepackage[figuresright]{rotating}
\usepackage{arydshln}
\usepackage{threeparttable}

\makeatletter

\@addtoreset{equation}{section} \makeatother
\newtheorem{Theorem}{Theorem}[section]
\newtheorem{Lemma}{Lemma}[section]

\newtheorem{Remark}{Remark}[section]
\newtheorem{Definition}{Definition}[section]
\newtheorem{Proposition}{Proposition}[section]
\newtheorem{Example}{Example}[section]
\makeatletter \@addtoreset{equation}{section} \makeatother
\makeatletter
\newcommand{\rmnum}[1]{\romannumeral#1}
\newcommand{\Rmnum}[1]{\expandafter\@slowromancap\romannumeral #1@}
\makeatother
\begin{document}
\title{New Constructions of Reversible DNA Codes}

\author{Xueyan Chen, Whan-Hyuk Choi, Hongwei Liu}
\author{Xueyan Chen\textsuperscript{1}, Whan-Hyuk Choi\textsuperscript{2}, Hongwei Liu\textsuperscript{1}}
\date{\textsuperscript{1} School of Mathematics and Statistics, Central China Normal University, Wuhan 430079, China\\
\textsuperscript{2}Department of Mathematics, Kangwon National University, 1 Gangwondaehakgil, Chuncheon 24341, Republic of Korea}
\maketitle

\insert\footins{\small{\it Email addresses}:
xueyanchen@mails.ccnu.edu.cn(Xueyan Chen), whchoi@kangwon.ac.kr (Whan-Hyuk Choi), hwliu@ccnu.edu.cn(Hongwei Liu)}

{\centering\section*{Abstract}}
 \addcontentsline{toc}{section}{\protect Abstract}
 \setcounter{equation}{0}
DNA codes have many applications, such as in data storage, DNA computing, etc. Good DNA codes have large sizes and satisfy some certain constraints. In this paper,  we present a new construction method for reversible DNA codes. We show that the DNA codes obtained using our construction method can satisfy some desired constraints and the lower bounds of the sizes of some DNA codes are better than the known results. We also give new lower bounds on the sizes of some DNA codes of lengths $80$, $96$ and $160$ for some fixed Hamming distance $d$.

\medskip
\noindent{\large\bf Keywords: }\medskip
Reversible DNA Codes, Reversible Composite Matrices, Group Codes, Group Rings\\
\noindent{\bf2010 Mathematics Subject Classification}: 15B33, 20C05, 94B05.

\section{Introduction}
The research on DNA codes originated in 1994 when Adleman solved a computationally difficult mathematical problem by introducing an algorithm using DNA strands and molecular biology tools~(\!\!\cite{b3}). Since then, many other applications for DNA codes have been discovered, such as digital media storage~(\!\!\cite{b9},~\cite{b12}), data encryption and combinatorial problems~(\!\!\cite{b4},~\cite{b31}), and cracking the DES cryptosystem ~(\!\!\cite{b6}).\par
DNA molecules consist of two complementary strands where each strand is a sequence of four different nucleotide bases, called adenine (A), cytosine(C), guanine (G) and thymine (T). It is well known that a DNA code satisfies the following constraints: ($\rmnum{1}$) Hamming distance constraint, ($\rmnum{2}$) reverse constraint, ($\rmnum{3}$) reverse-complement constraint, and ($\rmnum{4}$) fixed GC-content constraint. There are two additional significant constraints that have been adequately investigated in the literature in addition to the above ones: free from secondary structure; no continuous repetition of the identical sub-string(s)(see~\cite{b35},~\cite{b21},~\cite{b33},~\cite{b34}).\par

Some known methods for designing DNA codes that satisfy certain conditions include: linear constructions~(\!\!\cite{b20}), cyclic and extended cyclic constructions~(\!\!\cite{b1},~\cite{b2}) and constructions from reversible self-dual codes over $\mathbb{F}_{4}$~(\!\!\cite{b24}). In [33], the authors defined coterm polynomial and generating methods and $u^2$-module codes. These definitions were transformed into $G$-codes and group ring codes for DNA codes, which is the origin of using $G$-codes and group ring codes to construct DNA codes. In~\cite{b10},~\cite{b19},~\cite{b27}, the authors studied linear codes derived from group ring elements to generate DNA codes that satisfy certain constraints, and discovered some new lower bounds on the sizes of some DNA codes. Moreover in~\cite{b19}, the authors assumed that the image $\sigma(v)$ of the map $\sigma$ defined in~\cite{b23} could be divided into a block reversible matrix. Then for each block of $\sigma(v)$,  by mapping with different finite groups to construct a new matrix, and the authors generated some DNA codes using the obtained new matrices. Reversibility is a desired property for DNA codes, the authors (\!\!\cite{b19}) made a connection between reversible composite $G$-codes and DNA codes. They also presented an algorithm that determined the size of an $l$-conflict free DNA code(for some positive integer $l\leq n$) in which the codewords are free from secondary structures.\par

In this paper, we improve the method given in \cite{b19} for constructing reversible DNA codes using group code (for specific details, please see Theorem \ref{th.3.3} and Remark \ref{r.3.1}). We provide a group $G$ which ensures that our construction result is a general construction. Moreover, many of the codes we obtain have better lower bounds on the sizes of some known DNA codes and some of our DNA codes are new.\par

The rest of this article is organized as follows. In Section $2$, we recall some concepts in coding theory. We also introduce the definitions and properties of DNA codes, group rings, group codes, composite matrices and composite group codes. In Section $3$, we propose a new method for constructing reversible composite group codes and provide relevant proofs. In Section $4$, we give some specific groups and obtain some particular forms of composite matrices using the construction method in Section $3$. In Section $5$, we generate some reversible composite group codes on Magma using the composite matrices obtained in Section $4$, then obtain some reversible DNA codes, and compare their parameters and some lower bounds to previous ones. In Section $6$, we draw a conclusion about our article.

\section{Preliminaries}

\subsection{DNA Codes}

In this section, we recall some basic definitions of linear codes, DNA codes and some constraints of DNA codes.\par
Let $\mathbb{F}_{q}$ be the finite field of order $q$, where $q=p^{e}$ is a power of a prime number $p$. A code of length $n$ over $\mathbb{F}_{q}$ is a subset of  $\mathbb{F}_{q}^n$ and a linear code of length $n$ over $\mathbb{F}_{q}$ is a subspace of  $\mathbb{F}_{q}^n$, and we call an element of a linear code as a codeword. The Hamming distance $d(x, y)$ between two codewords is the number of coordinates in which $x$ and $y$ are distinct. The minimum Hamming distance $d$ of a linear code $\mathcal{C}$ is defined as min$\{d(x, y) | x \neq y, \forall x, y \in \mathcal{C} \}$. Let $S_{D_{4}}=\{A,T,C,G\}$ denote the set of nucleotides in DNA(Represented as adenine (A), cytosine (C), guanine (G) and thymine (T)). We use $\,\hat{}\,$ to denote the Watson-Crick complement of a nucleotide, i.e., $\hat{A}= T, \hat{T} = A, \hat{C} = G$ and $\hat{G}= C$. Let $S_{D_{4}}^{n}=\{(x_{1},\dots,x_{n})|x_{i}\in S_{D_{4}}\}$, for any $x=(x_{1},\dots,x_{n})\in S_{D_{4}}^{n}$, let $x^{r} = (x_{n}, \dots , x_{1})$ be the reverse of $x$ and $x^{c} = (\hat{x}_{1},\dots ,\hat{x}_{n})$ be the complement of $x$. Moreover, let ${x}^{rc} = (\hat{x}_{n},\dots, \hat{x}_{1})$ be the reverse and complement of $x$.

\begin{Definition}
Assume the notation is as given above. A DNA code $D$ of length $n$ is defined as a subset of $S_{D_{4}}^{n}$, such that $D$ satisfies some or all of the following constraints:\par
(\rmnum{1}) The Hamming distance constraint (HD):\par
 \qquad $d({\bf x}, {\bf y}) \geq d, \forall\,{\bf x}, {\bf y} \in D$, for some prescribed Hamming distance $d$.\par
(\rmnum{2}) The reverse constraint (RV):\par
\qquad $d({\bf x}^{r}, {\bf y}) \geq d, \forall\,{\bf x}, {\bf y} \in D$, including $x = y$ for some  prescribed Hamming distance $d$.\par
(\rmnum{3}) The reverse-complement constraint (RC):\par
\qquad $d({\bf x}^{rc}, {\bf y}) \geq d, \forall\,{\bf x}, {\bf y} \in D$, including $x = y$ for some  prescribed Hamming distance $d$.\par
(\rmnum{4}) The fixed GC-content constraint (GC):\par
\qquad The set of codewords with length $n$, distance $d$ and GC weight $w_{\bf x_{\text{DNA}} }$, where $w_{\bf x_{\text{DNA}} }$ is the total number of G's and C's present in the DNA strand, i.e.,\par
\begin{center}
 $w_{\bf x_{\text{DNA}} }= |\{x_{i} | {\bf x} = (x_{i}), x_{i}=C ~\text{or}~ G\}|$.
\end{center}

\end{Definition}
Note that the first constraint must be satisfied among the four constraints mentioned above, while the remaining constraints can satisfy some combination forms.
In this paper, the fixed $GC$-content of a DNA code $D$ is simply $\lfloor \frac{n}{2}\rfloor$, where $n$ is the length of the code.\par
 Let $\mathbb{F}_{4} = \{0, 1, \omega, \omega^{2}\}$, where $\omega^{2}=\omega+1$, be the finite field of order $4$. Let $\eta$ be a bijective correspondence between $\mathbb{F}_{4}$ and the DNA alphabet $S_{D_{4}}= \{A, T, C, G\}$ given by
\begin{center}
$\eta : \mathbb{F}_{4} \longrightarrow S_{D_{4}}$,
\end{center}
with $\eta(0) = A, \eta(1) = T, \eta(\omega) = C$ and $\eta(\omega^{2}) = G$. The bijection $\eta$ can be extended from $\mathbb{F}^n_4$ to $S^n_{D_{4}}$ naturally. Therefore, a DNA code can be identified with a code over $\mathbb{F}_{4}$. \par
We denote the complete weight enumerator of a code $\mathcal{C}$ over $\mathbb{F}_{4}$ by
\begin{center}
$\text{CWE}_{\mathcal{C}}(X_{1},X_{2}, X_{3}, X_{4}) =\sum\limits_{{\bf c} \in \mathcal{C}}X_{1}^{n_{0}({\bf c})}X_{2}^{n_{1}({\bf c})}X_{3}^{n_{\omega}({\bf c})}X_{4}^{n_{\omega^{2}}({\bf c})}$,
 \end{center}
where $n_{s}({\bf c})$ denotes the number of occurrences of $s$ in a codeword ${\bf c}$. We identify the complete weight enumerator of a DNA code $D$ with that of a code $\mathcal{C}$ over $\mathbb{F}_{4}$, where $D = \eta(\mathcal{C})$. The $GC$-weight of a codeword ${\bf c} \in \mathcal{C}$ is the sum of $n_{\omega}({\bf c})$ and $n_{\omega^{2}}({\bf c})$. Therefore, if we let
\begin{center}
$\text{GCW}_{\mathcal{C}}(X_{1}, X_{2}) = \text{CWE}_{\mathcal{C}}(X_{1}, X_{1}, X_{2}, X_{2})$,
\end{center}
then $\text{GCW}_{\mathcal{C}}(X_{1}, X_{2})$ is the $GC$-weight enumerator of a code $\mathcal{C}$, where the coefficient of $X_2^{i}$ is the same as the number of codewords with GC-weight $i$.\par

\begin{Remark}
(1)~In the later calculation, the number of codewords in the DNA code of length $n$ that satisfies the GC-content of $\lfloor \frac{n}{2}\rfloor$ corresponds to the coefficient of $X_2^{\lfloor \frac{n}{2}\rfloor }$ in the $GC$-weight enumerator.\par
(2)~In addition to the constraints mentioned above, there are two other critical constraints that have been frequently researched in the literature, since these two constraints are not the focus of this work, we omit it here, please refer to (\!\!\cite{b21},~\cite{b35},~\cite{b33},~\cite{b34}) for more information.
\end{Remark}

Let $A_{4}^{R}(n, d)$ denote the maximum size of a DNA code for a given distance $d$ and length $n$ that satisfies the HD and RV constraints. Let $A_{4}^{RC}(n, d)$ be the maximum size of a DNA code of length $n$ satisfying the HD and RC constraints for a given $d$, $A_{4}^{GC}(n,d, \lfloor \frac{n}{2}\rfloor)$ be the maximum size of a DNA code of length $n$ satisfying the HD constraint for a given $d$ with a constant $GC$-weight $\lfloor \frac{n}{2}\rfloor$, and $A_{4}^{RC,GC}(n, d,\lfloor \frac{n}{2}\rfloor)$ the maximum size of a DNA code of length $n$ satisfying the HD and RC constraints for a given $d$ with a constant $GC$-weight $\lfloor \frac{n}{2}\rfloor$. In~\cite{b29}, for even $n$, the following equality is given:
\begin{equation}\label{eq:0}
A_{4}^{RC}(n,d)=A_{4}^{R}(n,d).
\end{equation}

A DNA code is called a {\it better DNA code} if it meets all the above constraints with better parameters, or if it is a DNA code with larger size under the same parameters. In this work, we construct some new DNA codes to improve the upper bounds of sizes, that is, these DNA codes are better DNA codes.\par

\subsection{Circulant Matrices, Group Rings and Group Codes}

In this section, we recall some definitions of several special forms of matrices that we use later in this paper. We also provide some necessary definitions for group rings, which will be used to construct the desired codes.\par
\begin{Definition}
Let $R$ be a finite ring, let $l$ be a positive integer. An $l$\text{-}circulant matrix is a matrix where each row is shifted $l$ elements to the right relative to the preceding row. We label the $l$\text{-}circulant matrix as $l\text{-}circ(\alpha_{1},\alpha_{2},\dots,\alpha_{n})$, where $\alpha_{i}$ are the ring elements appearing in the first row. Specifically, the $l\text{-}circ(\alpha_{1},\alpha_{2},\dots,\alpha_{n})$ matrix  is:
$$l\text{-}circ(\alpha_{1},\alpha_{2},\dots,\alpha_{n})
=\begin{pmatrix}
\alpha_{1}&\alpha_{2}&\alpha_{3}&\dots&\alpha_{n-1}&\alpha_{n}\\
\alpha_{n-l+1}&\alpha_{n-l+2}&\alpha_{n-l+3}&\dots&\alpha_{n-l-1}&\alpha_{n-l}\\
\alpha_{n-2l+1}&\alpha_{n-2l+2}&\alpha_{n-2l+3}&\dots&\alpha_{n-2l-1}&\alpha_{n-2l}\\
\vdots&\vdots&\vdots&&\vdots&\vdots\\
\alpha_{l+1}&\alpha_{l+2}&\alpha_{l+3}&\dots&\alpha_{l-1}&\alpha_{l}\\
\end{pmatrix}_{n \times n}
.$$
\end{Definition}
Thus, the circulant matrix  is $1\text{-}circ(\alpha_{1},\alpha_{2},\dots,\alpha_{n})$-circulant matrix. Let $A=(a_{ij})_{n\times n}$ be a square matrix of degree $n$ over $R$. Let $A^{T}$ denote the transpose of $A$. The flip of $A$, denoted by $A^{F}$, is $A^{F} = (a_{n-i+1,n-j+1})_{n\times n}$.\par
We shall now give the standard definition of group rings.\par
\begin{Definition}
Let $G$ be a finite group of order $n$ and let $R$ be a finite ring. Let $RG=\{v=\sum\limits_{i=1}^{n}\alpha_{g_i}g_{i}|\alpha_{g_i} \in R, g_{i} \in G\}$.  Let $\sum\limits_{i=1}^{n}\alpha_{g_i}g_{i}, \sum\limits_{i=1}^{n}\beta_{g_i}g_{i} \in RG$, define
\begin{center}
$\sum\limits_{i=1}^{n}\alpha_{g_i}g_{i}+\sum\limits_{i=1}^{n}\beta_{g_i}g_{i}=
\sum\limits_{i=1}^{n}(\alpha_{g_i}+\beta_{g_i})g_{i}. $\\
$(\sum\limits_{i=1}^{n}\alpha_{g_i}g_{i})(\sum\limits_{j=1}^{n}\beta_{g_j}g_{j})=
\sum\limits_{i,j}\alpha_{g_i}\beta_{g_j}g_{i}g_{j}
=\sum\limits_{k=1}^{n}(\sum\limits_{g_{i}g_{j}=g_{k}}\alpha_{g_i}\beta_{g_j})g_{k}.$
\end{center}
Then RG is a ring (called as group ring) under the above two operations. If the ring $R$ is a field then $RG$ is said to be a group algebra.\par
\end{Definition}

In~\cite{b13}, Dougherty {\it et al} derived a matrix using the elements in a group ring, and then used this matrix to generate linear codes (i.e., group codes). Now we  first recall the construction process of group codes. The following matrix construction was given by Hurley in\cite{b23}. The same matrix construction was used to study group codes over Frobenius rings in\cite{b13}.\par
Let $G$ be a finite group, let $\{g_{1},g_{2}, \dots , g_{n}\}$ be a fixed listing of the elements of $G$ and $R$ be a finite commutative Frobenius ring. Let $M(G)$ be the following matrix:
\begin{eqnarray}\label{eq.2.2}
M(G)=\begin{pmatrix}
g_{1}^{-1}g_{1}&g_{1}^{-1}g_{2}&\dots&g_{1}^{-1}g_{n}\\
g_{2}^{-1}g_{1}&g_{2}^{-1}g_{2}&\dots&g_{2}^{-1}g_{n}\\
\vdots&\vdots&&\vdots\\
g_{n}^{-1}g_{1}&g_{n}^{-1}g_{2}&\dots&g_{n}^{-1}g_{n}\\
\end{pmatrix}_{n\times n}.
\end{eqnarray}

Let $v=\sum\limits_{i=1}^{n}\alpha_{g_i}g_{i}\in RG.$ Let $M(RG,v)$ denote the $R$-matrix constructed from $M(G)$ and $v$ as follows:

\begin{eqnarray}\label{eq.2.3}
M(RG,v)=\begin{pmatrix}
\alpha_{g_{1}^{-1}g_{1}}&\alpha_{g_{1}^{-1}g_{2}}&\dots  &\alpha_{g_{1}^{-1}g_{n}}\\
\alpha_{g_{2}^{-1}g_{1}}&\alpha_{g_{2}^{-1}g_{2}}&\dots  &\alpha_{g_{2}^{-1}g_{n}}\\
\vdots   &\vdots & &\vdots\\
\alpha_{g_{n}^{-1}g_{1}}&\alpha_{g_{n}^{-1}g_{2}}&\dots  &\alpha_{g_{n}^{-1}g_{n}}\\
\end{pmatrix}_{n \times n}.
\end{eqnarray}

Let  $$\sigma: RG\longrightarrow M_{n}(R), ~~v\mapsto \sigma(v)\triangleq M(RG,v).$$

In~\cite{b23}, the author proved that $\sigma$ is an injection ring homomorphism.

\begin{Definition}(\!\!\cite{b13})
Assume the notation is as given above. Let
\begin{eqnarray}\label{eq.2.4}
\mathcal{C}(v)=\langle  \sigma(v)\rangle
\end{eqnarray}
be a linear code generated by $\sigma(v)$ over $R$, where $\langle\sigma(v)\rangle$ is a submodule of $R^n$ generated by the rows of $\sigma(v)$. The code $\mathcal{C}(v)$ is also called a $G$-code.
\end{Definition}
Moreover, in~\cite{b13}, Dougherty {\it et al} proved that the group code constructed in this way is a left ideal of $RG$ under the following corresponding $$\Psi:R^{n}\longrightarrow RG, (\alpha_{1},\dots,\alpha_{n}) \mapsto \sum\limits_{i=1}^{n} \alpha_{i}g_{i}.$$
Thus the resulting group code has the group $G$ as a subgroup of its automorphism group. For more details on group codes generated from group rings please see~(\!\!\cite{b13}). From now on, we refer to $G$-codes, which means the codes given in Equation~\eqref{eq.2.4}.\par

\subsection{Composite Matrices and Composite Group Codes }
We now recall the composite matrix construction which was first given in~\cite{b17}.\par
Let $R$ be a finite commutative Frobenius ring. Let $\{g_{1}, g_{2},\dots, g_{n}\}$ be a fixed listing of the elements of $G$. Let $v=\sum\limits_{i=1}^{n} \alpha_{g_{i}}g_{i} \in RG$. Suppose  $1<r<n,  r~|~n $, and let $H_{i}$ be an arbitrary group of order $r$ with $i\in\{1, 2, \dots , \frac{n^2}{r^{2}}\}$. Let $\{h_{i1}, h_{i2},\dots , h_{ir}\}$ be a fixed listing of the elements of $H_{i}$. Let $G_{r}$ be a subset of $G$ containing $r$ distinct elements of $G$. Let $l$ be a positive integer with $1\leq l\leq \frac{n^2}{r^{2}}$.  For $i,l\in  \{1, 2, \dots , \frac{n^2}{r^2} \}$, define a bijection $\phi_{i,l}$ as follows:
$$
\phi_{i,l}: H_{i} \rightarrow G_{r}, h_{it} \mapsto g_{j}^{-1}g_{k+t-1}, ~ \forall ~ 1\le t\le r,
$$
where the values of the pairs $(j,k)$, and  $l$  are defined as follows:
$$
(j,k)  =\begin{cases}
(1,(l-1)r+1),&1\leq l \leq \frac{n}{r};\\
(r+1,(l-\frac{n}{r}-1)r+1),&\frac{n}{r}+1\leq l \leq \frac{2n}{r};\\
\cdots & \cdots \\
((\frac{n}{r}-1)r+1,[l-(\frac{n}{r}-1)\frac{n}{r}-1]r+1), &(\frac{n}{r}-1)\frac{n}{r}+1\leq l \leq \frac{n^2}{r^2}.
\end{cases}
$$
\begin{Definition}\label{def.2.5}
Assume the notation is as given above. The following matrix
$$
\Omega(v)=\begin{pmatrix}
A_{1}&A_{2}&\dots  &A_{\frac{n}{r}}\\
A_{\frac{n}{r}+1} &A_{\frac{n}{r}+2}&\dots&A_{\frac{2n}{r}}\\
\vdots   &\vdots & &\vdots\\
A_{(\frac{n}{r}-1)\frac{n}{r}+1} &A_{(\frac{n}{r}-1)\frac{n}{r}+2}&\dots &A_{\frac{n^{2}}{r^{2}}}
\end{pmatrix}_{n\times n}
$$
is called a composite matrix, where at least one block has the following form:
$$
A_{l}=\begin{pmatrix}
\alpha_{g_{j}^{-1}g_{k}}&\alpha_{g_{j}^{-1}g_{k+1}}&\dots&\alpha_{g_{j}^{-1}g_{k+(r-1)}}\\
\alpha_{\phi_{i,l}(h_{i2}^{-1}h_{i1})}&\alpha_{\phi_{i,l}(h_{i2}^{-1}h_{i2})}
&\dots&\alpha_{\phi_{i,l}(h_{i2}^{-1}h_{ir})}\\
\vdots   &\vdots            &      &\vdots\\
\alpha_{\phi_{i,l}(h_{ir}^{-1}h_{i1})}&\alpha_{\phi_{i,l}(h_{ir}^{-1}h_{i2})}
&\dots&\alpha_{\phi_{i,l}((h_{ir}^{-1}h_{ir})}\\
\end{pmatrix}_{r \times r}\triangleq A'_{l},
$$
and the other blocks have the following form:\par
$$
A_{l}=\begin{pmatrix}
\alpha_{g_{j}^{-1}g_{k}}&\alpha_{g_{j}^{-1}g_{k+1}}&\dots&\alpha_{g_{j}^{-1}g_{k+(r-1)}}\\
\alpha_{g_{j+1}^{-1}g_{k}}&\alpha_{g_{j+1}^{-1}g_{k+1}}&\dots&\alpha_{g_{j+1}^{-1}g_{k+(r-1)}}\\
\vdots   &\vdots          &      &\vdots\\
\alpha_{g_{j+(r-1)}^{-1}g_{k}}&\alpha_{g_{j+(r-1)}^{-1}g_{k+1}}&\dots&\alpha_{g_{j+(r-1)}^{-1}g_{k+(r-1)}}\\
\end{pmatrix}_{r\times r}.
$$
\end{Definition}


\begin{Definition}(\!\!\cite{b18})
Assume the notation is as given above. Let
\begin{eqnarray}\label{eq.2.5}
\mathcal{D}(v)=\langle  \Omega(v) \rangle
\end{eqnarray}
be a linear code generated by $\Omega(v)$ over $R$, where $\langle  \Omega(v) \rangle$ is a submodule of $R^n$ generated by the rows of $\Omega(v)$. The code $\mathcal{D}(v)$ is also called a composite $G$-code.
\end{Definition}

The main advantage of matrix $\Omega(v)$ is that the codes generated by $\Omega(v)$ can obtain parameters that cannot be obtained from the codes generated by $\sigma(v)$. For several instances of code families constructed from the composite matrix $\Omega(v)$, please see~(\!\!\cite{b14},~\cite{b15},~\cite{b18}). From now on, when we refer to composite $G$-codes, we mean codes given in Equation \eqref{eq.2.5}.\par

\section{Reversible Composite Group Codes}
\subsection{Reversible Group Codes}
We now review a crucial finding in~\cite{b10}, which shows that it is possible to create reversible $G$-codes for specific groups. We first provide relevant definitions from~\cite{b10}.\par

\begin{Definition}
A code $\mathcal{C}$ is said to be reversible of index $k$ if $\mathbf{a}_{i}$ is a vector of length $k$ ( i.e.,  $\mathbf{a}_{i}=(a_{i1},a_{i2},\dots,a_{ik})$ ) and ${\bf c} =(\mathbf{a}_{1},\mathbf{a}_{2},\dots ,\mathbf{a}_{s-1},\mathbf{a}_{s}) \in \mathcal{C}$ implies that ${\bf c}^{r} =(\mathbf{a}_{s},\mathbf{a}_{s-1},\dots ,\mathbf{a}_{2},\mathbf{a}_{1}) \in \mathcal{C}.$
\end{Definition}

\begin{Example}
Let $\mathcal{C}=\{a(1,1,\omega^2,0,0,0,1,1,\omega^2)+b(0,0,0,1,1,1,\omega,\omega,1)\,|\, a,b \in \mathbb{F}_4\}$ be a $[9,2]$-linear code over ~$\mathbb{F}_{4}$. For any codeword ~${\bf c} \in \mathcal{C}$, ${\bf c}=(a,a,a\omega^2,b,b,b,a+b\omega,a+b\omega,a\omega^2+b)$, where $a,b \in \mathbb{F}_4$. It is easy to check that ~${\bf c}^{r}=(a+b\omega,a+b\omega,a\omega^2+b,b,b,b,a,a,a\omega^2) =(a^{'},a^{'},a^{'}\omega^2,b^{'},b^{'},b^{'},a^{'}+b^{'}\omega,a^{'}+b^{'}\omega,a^{'}\omega^2+b^{'})\in \mathcal{C}$, where $ a^{'}=a+b\omega, ~b^{'}=b$. Therefore ~$\mathcal{C}$ is a reversible code of index $3$.
\end{Example}

Now we give a special order of elements in a group. Let $G$ be a finite group of order $n=2l$ and let $T =\{e, t_{1},t_{2},\dots,t_{l-1}\}$ be a subgroup of index $2$ in $G$. Let $\beta \in G\backslash T$ be of order $2$. We list the elements of $G=\{g_{1}, g_{2},\dots,g_{n}\}$ as follows:\par
\begin{eqnarray}\label{eq.3.1}
\{e,t_{1},t_{2},\dots,t_{l-1},\beta t_{l-1},\beta t_{l-2},\dots,\beta t_{1},\beta \}.
\end{eqnarray}

The following result was proved in~\cite{b10}.\par
\begin{Theorem}
(\!\!\cite{b10}) Let $R$ be a finite ring. Let $G$ be a finite group of order $n=2l$ and let $T =\{e, t_{1},t_{2},\dots,t_{l-1}\}$ be a subgroup of index $2$ in $G$. Let $\beta \in G\backslash T$ be of order $2$. List the elements of $G$ as in \eqref{eq.3.1}, then any linear $G$-code in $R^{n}$ (a left ideal in $RG$) is a
reversible code of index $1$.
\end{Theorem}
Reversibility is a desired property of DNA codes. In~\cite{b10}, Cengellenmis {\it et al} made a connection between reversible $G$-codes and DNA codes. We intend to build connections between reversible composite $G$-codes and DNA codes inspired by~\cite{b10}. In the following subsection, we will provide our construction methods.\par

\subsection{Reversible Composite Group Codes}
Let $G$ be a finite group of order $n$ such that  $4~|~ n$, and let $H_{i}$, where $i = \{1, 2, . . . , s \}$, be finite groups of even order $r$ such that $r$ is a factor of $n$ with $1<r<n$, where $s$ is a positive integer. Let $R$ be a finite commutative Frobenius ring, and $v=\sum\limits_{i=1}^n\alpha_ig_i$. We partition the matrix $\sigma(v)$ into an $\frac{n}{r} \times \frac{n}{r}$ block matrix as follows:

$$
\sigma(v)=\begin{pmatrix}
A_{1}&A_{2}&\dots  &A_{\frac{n}{r}}\\
A_{\frac{n}{r}+1} &A_{\frac{n}{r}+2}&\dots&A_{\frac{2n}{r}}\\
\vdots   &\vdots & &\vdots\\
A_{(\frac{n}{r}-1)\frac{n}{r}+1} &A_{(\frac{n}{r}-1)\frac{n}{r}+2} &\dots  &A_{\frac{n^{2}}{r^{2}}}
\end{pmatrix}_{n \times n},
$$
where each block $A_i$ belongs to $M_r(R)$.

By Definition~\ref{def.2.5}, $\Omega(v)$ can be obtained from $\sigma(v)$. In this subsection, we will add some restrictions on $\Omega(v)$ to ensure that the resulting matrix is reversible, and it will be employed to generate reversible composite group codes.

In general, $\sigma(v)$ is not a reversible block matrix. The following theorem prove that $\sigma(v)$ can be reversible by suitable choices of  a group $G$.

\begin{Theorem}\label{th.3.2}
Let $G = \langle x, y ~|~ x^{r}= y^{\frac{n}{r}} = 1, xy = yx \rangle = C_{r}\times C_{\frac{n}{r}}$, where $n$ is a positive integer such that $4~|~n$, $r$ is a factor of $n$ with $1<r<n$ and $\frac{n}{r}$ is even. Let
{\small \begin{eqnarray}\label{eq.3.2}
v=\sum\limits_{i=0}^{r-1}[\alpha_{g_{i+1}}x^{i}+\alpha_{g_{i+(\frac{n}{r}-1)r+1}}x^{i}y^{\frac{n}{2r}}]
+\sum\limits_{k=1}^{\frac{n}{2r}-1}\sum\limits_{j=0}^{r-1}[\alpha_{g_{j+(2k-1)r+1}}x^{j}y^{k}
+\alpha_{g_{j+2kr+1}}x^{j}y^{\frac{n}{r}-k}].
\end{eqnarray}}
Then the following block matrix:
$$
\sigma(v)=\begin{pmatrix}
A_{1}&A_{2}&\dots  &A_{\frac{n}{r}}\\
A_{\frac{n}{r}+1} &A_{\frac{n}{r}+2}&\dots&A_{\frac{2n}{r}}\\
\vdots   &\vdots & &\vdots\\
A_{(\frac{n}{r}-1)\frac{n}{r}+1} &A_{(\frac{n}{r}-1)\frac{n}{r}+2} &\dots  &A_{\frac{n^{2}}{r^{2}}}
\end{pmatrix}_{n \times n}
$$
is block reversible.
\end{Theorem}

\begin{proof}
Note that $G = \langle x, y~|~ x^{r}= y^{\frac{n}{r}} = 1, xy = yx \rangle$ and
$$v =\sum\limits_{i=0}^{r-1}[\alpha_{g_{i+1}}x^{i}+\alpha_{g_{i+(\frac{n}{r}-1)r+1}}x^{i}y^{\frac{n}{2r}}]
+\sum\limits_{k=1}^{\frac{n}{2r}-1}\sum\limits_{j=0}^{r-1}[\alpha_{g_{j+(2k-1)r+1}}x^{j}y^{k}
+\alpha_{g_{j+2kr+1}}x^{j}y^{\frac{n}{r}-k}].$$
We have
\begin{eqnarray*}
\begin{aligned}
G=\{g_{1},g_{2},\dots,g_{n}\}=\Big \{&1,x,\dots,x^{r-1},\quad y,xy,\dots,x^{r-1}y,\\
 &y^{\frac{n}{r}-1},xy^{\frac{n}{r}-1},\dots,x^{r-1}y^{\frac{n}{r}-1},\quad \dots,\quad y^{\frac{n}{2r}-1},xy^{\frac{n}{2r}-1},\dots,x^{r-1}y^{\frac{n}{2r}-1},\\
 &y^{\frac{n}{2r}+1},xy^{\frac{n}{2r}+1},\dots,x^{r-1}y^{\frac{n}{2r}+1},\quad y^{\frac{n}{2r}},xy^{\frac{n}{2r}},\dots,x^{r-1}y^{\frac{n}{2r}} \Big\}.
 \end{aligned}
\end{eqnarray*}
Thus
\begin{eqnarray*}
\begin{aligned}
\{g_{1}^{-1},g_{2}^{-1},\dots,g_{n}^{-1}\}=\Big\{&1,x^{r-1},\dots,x,\quad y^{\frac{n}{r}-1},x^{r-1}y^{\frac{n}{r}-1},\dots,xy^{\frac{n}{r}-1},\\
&y,x^{r-1}y,\dots,xy,\quad \dots, \quad y^{\frac{n}{2r}+1},x^{r-1}y^{\frac{n}{2r}+1},\dots,xy^{\frac{n}{2r}+1},\\
&y^{\frac{n}{2r}-1},x^{r-1}y^{\frac{n}{2r}-1},\dots,xy^{\frac{n}{2r}-1},\quad y^{\frac{n}{2r}},x^{r-1}y^{\frac{n}{2r}},\dots,xy^{\frac{n}{2r}} \Big\}.
\end{aligned}
\end{eqnarray*}
By Equation \eqref{eq.2.2}, we have
$$
M(G)=\begin{pmatrix}
M_{0}&M_{1}&M_{\frac{n}{r}-1}&\dots  &M_{\frac{n}{2r}-1}&M_{\frac{n}{2r}+1}&M_{\frac{n}{2r}}\\
M_{\frac{n}{r}-1} &M_{0}&M_{\frac{n}{r}-2}&\dots&M_{\frac{n}{2r}-2}&M_{\frac{n}{2r}}&M_{\frac{n}{2r}-1}\\
M_{1}&M_{2}&M_{0}&\dots&M_{\frac{n}{2r}}&M_{\frac{n}{2r}+2}&M_{\frac{n}{2r}+1}\\
\vdots   &\vdots &\vdots & &\vdots&\vdots&\vdots\\
M_{\frac{n}{2r}+1}&M_{\frac{n}{2r}+2} &M_{\frac{n}{2r}} &\dots&M_{0}&M_{2}&M_{1}\\
M_{\frac{n}{2r}-1}&M_{\frac{n}{2r}} &M_{\frac{n}{2r}-2} &\dots&M_{\frac{n}{r}-2}&M_{0}&M_{\frac{n}{r}-1}\\
M_{\frac{n}{2r}}&M_{\frac{n}{2r}+1} &M_{\frac{n}{2r}-1} &\dots&M_{\frac{n}{r}-1}&M_{1}&M_{0}\\
\end{pmatrix},
$$
where
$$
M_{i}=\begin{pmatrix}
y^{i}&xy^{i}&x^2y^{i}&\dots&x^{r-2}y^{i}&x^{r-1}y^{i}\\
x^{r-1}y^{i}&y^{i}&xy^{i}&\dots&x^{r-3}y^{i}&x^{r-2}y^{i}\\
x^{r-2}y^{i}&x^{r-1}y^{i}&y^{i}&\dots&x^{r-4}y^{i}&x^{r-3}y^{i}\\
\vdots&\vdots&\vdots&&\vdots&\vdots\\
x^2y^{i}&x^3y^{i}&x^4y^{i}&\dots&y^{i}&xy^{i}\\
xy^{i}&x^2y^{i}&x^3y^{i}&\dots&x^{r-1}y^{i}&y^{i}\\
\end{pmatrix}_{r \times r},
$$
then $M(G)$ is a block reversible matrix.\par
By Equation \eqref{eq.2.3}, we have

$$
\sigma(v)=\begin{pmatrix}
A_{0}&A_{1}&A_{\frac{n}{r}-1}&\dots  &A_{\frac{n}{2r}-1}&A_{\frac{n}{2r}+1}&A_{\frac{n}{2r}}\\
A_{\frac{n}{r}-1} &A_{0}&A_{\frac{n}{r}-2}&\dots&A_{\frac{n}{2r}-2}&A_{\frac{n}{2r}}&A_{\frac{n}{2r}-1}\\
A_{1}&A_{2}&A_{0}&\dots&A_{\frac{n}{2r}}&A_{\frac{n}{2r}+2}&A_{\frac{n}{2r}+1}\\
\vdots   &\vdots &\vdots & &\vdots&\vdots&\vdots\\
A_{\frac{n}{2r}+1}&A_{\frac{n}{2r}+2} &A_{\frac{n}{2r}} &\dots&A_{0}&A_{2}&A_{1}\\
A_{\frac{n}{2r}-1}&A_{\frac{n}{2r}} &A_{\frac{n}{2r}-2} &\dots&A_{\frac{n}{r}-2}&A_{0}&A_{\frac{n}{r}-1}\\
A_{\frac{n}{2r}}&A_{\frac{n}{2r}+1} &A_{\frac{n}{2r}-1} &\dots&A_{\frac{n}{r}-1}&A_{1}&A_{0}\\
\end{pmatrix}_{n\times n}
,$$
is also a block reversible matrix.
\end{proof}

In the following, by suitable choices of  group $G$, we can obtain that  $\sigma(v)$ is a block reversible matrix. we construct $\Omega^{\ast}(v)$ by fixing each block in $\sigma(v)$ to be of the form:
$$
A_{l}^{'}=\begin{pmatrix}
\alpha_{g_{j}^{-1}g_{k}}&\alpha_{g_{j}^{-1}g_{k+1}}&\dots&\alpha_{g_{j}^{-1}g_{k+(r-1)}}\\
\alpha_{\phi_{i,l}(h_{i2}^{-1}h_{i1})}&\alpha_{\phi_{i,l}(h_{i2}^{-1}h_{i2})}&\dots&\alpha_{\phi_{i,l}(h_{i2}^{-1}h_{ir})}\\
\vdots   &\vdots &      &\vdots\\
\alpha_{\phi_{i,l}(h_{ir}^{-1}h_{i1})}&\alpha_{\phi_{i,l}(h_{ir}^{-1}h_{i2})}&\dots&\alpha_{\phi_{i,l}(h_{ir}^{-1}h_{ir})}\\
\end{pmatrix}_{r \times r},
$$
where the elements of the groups $H_{i}$ are listed as form Equation \eqref{eq.3.1} and for each same block corresponding to the reverse rows in the block matrix $\sigma(v)$, we use the same group $H_{i}, i \in \{1,2,\dots,s\}$ to construct $A'_l$.\par

\begin{Theorem}\label{th.3.3}
Let $\Omega^{\ast}(v)$ be the composite matrix defined above and let
\begin{eqnarray}\label{eq.3.3}
\mathcal{D}^{\ast}(v) = \langle \Omega^{\ast}(v)\rangle
\end{eqnarray}
be a composite $G$-code code generated by $\Omega^{\ast}(v)$ over $R$. Then any linear composite $G$-code in $R^{n}$ generated by~\eqref{eq.3.3} is a reversible code of index $1$.
\end{Theorem}

\begin{Remark}\label{r.3.1}
Theorem \ref{th.3.3} is an improvement of the main conclusion in reference \cite{b19}. During our research, we observed that if the finite group $H$ used for transforming each block of $\sigma(v)$ does not follow certain rules, the DNA code generated by the method given in \cite{b19} may not be reversible. Therefore, we changed some conditions of this method and obtained a general construction.
\end{Remark}

Now we will provide a counter example we obtained.\par
\begin{Example}
Assume $G= \langle x,y~|~x^6=y^2=1,xy=yx\rangle \cong C_2\times C_6$, then $G=\{g_1, g_2, \dots , g_{12}\}$, where $g_i = x^{i-1}$ for $i = 1$ to $i = 6$ and $g_{i+6} = x^{i-1}y$ for $i = 1$ to $i = 6$. Let $R=\mathbb{F}_4$, $v=\sum_{i=0}^{5}\alpha_{g_{i+1}}x^{i}+\alpha_{g_{i+7}}x^{i}y \in R(C_2\times C_6)$. Then
$$\sigma(v)=\begin{pmatrix}
A&B\\
B&A
\end{pmatrix},$$
where $A=cir(\alpha_{g_1},\alpha_{g_2},\dots,\alpha_{g_6})$, $B=cir(\alpha_{g_7},\alpha_{g_8},\dots,\alpha_{g_{12}})$. \par
Let $H_1=\langle a,b~|~a^3=b^2=1,a^b=a^{-1}\rangle$ be a dihedral group of order $6$, $H_2=\langle c~|~c^6=1\rangle$ be a cyclic group of order $6$. According to the construction method in reference~\cite{b19}, we can use different groups $H_i$ to construct $A_l^{'}$($1\leq l \leq4$ ) for the four blocks of $\sigma(v)$.
 $$
\sigma(v)=\begin{pmatrix} 	
A&B\\
B&A
\end{pmatrix}
\longrightarrow
\Omega^{\ast}(v)=\begin{pmatrix}
A_{1}^{'}&A_{2}^{'}\\
 A_{3}^{'}&A_{4}^{'}
 \end{pmatrix}.
$$
therefore, we use group $H_1$ to construct $A_1^{'}$, $A_3^{'}$ and group $H_2$ to construct $A_2^{'}$, $A_4^{'}$. Then we obtained the following matrix:

$$\Omega^{\ast}(v)
=\addtocounter{MaxMatrixCols}{10}
\left(
\begin{array}{cccccc:cccccc}
1&0&0& \omega&0&0&   \omega&\omega^2&\omega& 1&\omega&\omega^2\\
0&1&0& 0&\omega&0&   \omega&\omega&\omega^2& \omega&\omega^2&1\\
0&0&1& 0&0&\omega&   \omega^2&\omega&\omega& \omega^2&1&\omega\\
\omega&0&0& 1&0&0&   \omega&1&\omega^2& \omega&\omega&\omega^2\\
0&\omega&0& 0&1&0&   1&\omega^2&\omega& \omega^2&\omega&\omega\\
0&0&\omega& 0&0&1&   \omega^2&\omega&1& \omega&\omega^2&\omega\\
\hdashline
\omega&\omega^2&\omega& 1&\omega&\omega^2&   1&0&0& \omega&0&0\\
\omega&\omega&\omega^2& \omega^2&1&\omega&   0&1&0& 0&0&\omega\\
\omega^2&\omega&\omega& \omega&\omega^2&1&   0&0&1& 0&\omega&0\\
1&\omega^2&\omega& \omega&\omega&\omega^2&   0&\omega&0& 1&0&0\\
\omega&1&\omega^2& \omega^2&\omega&\omega&   \omega&0&0& 0&1&0\\
\omega^2&\omega&1& \omega&\omega^2&\omega&   0&0&\omega& 0&0&1
\end{array}
\right).
$$

Let $\mathcal{C}=\langle \Omega^{\ast}(v) \rangle$. It is easy to check that codeword ${\bf c}= ( 0,0,0,0,0,0,0,0,0,\omega,0,\omega)\\ \in \mathcal{C}$, but ${\bf c^{r}}= ( \omega,0,\omega,0,0,0,0,0,0,0,0,0) \notin \mathcal{C}$. Therefore it is not a reversible composite $G$-code. The DNA code obtained using this composite $G$-code is also not reversible.
\end{Example}

To prove Theorem \ref{th.3.3}, we first give the following two lemmas:
\begin{Lemma}\label{lemma.3.1}
(\!\!\cite{b19}) Let $r$ be an even integer and let $A_{l}^{'}$ be the $r\times r$ matrix with the elements of the group $H_{i}$ are listed as form Equation \eqref{eq.3.1}. Then the reverse of each row of $A_{l}^{'}$ is in $A_{l}^{'}$.
\end{Lemma}

Based on the property given in Lemma~\ref{lemma.3.1}, we have the following result:
\begin{Lemma}\label{lemma.3.2}
Let $r$ be an even integer, $A_{l}^{'}$ be the $r\times r$ matrix with the elements of the group $H_{i}$ being listed as in Equation \eqref{eq.3.1} and $\mathbf{a}$ be a vector of length $r$. Then $(\mathbf{a}A_{l}^{'})^{r}=\mathbf{a}PA_{l}^{'}$, where $P$ is a permutation matrix related to $A_{l}^{'}$, and for all finite group $H_{i}$ with the form as Equation \eqref{eq.3.1}, the permutation matrix $P$ are the same.
\end{Lemma}

\begin{proof}
By Lemma~\ref{lemma.3.1}, we know that the reverse of each row of $A_{l}^{'}$ is in $A_{l}^{'}$. We suppose that in $A_{l}^{'}$ the reverse of the $i_{k}$th row is the $j_{k}$th row, where $k \in \{1,2,\dots,\frac{r}{2}\}$. Let $\mathbf{a}=(a_{1},\dots,a_{r})$, $P=E(i_{1},j_{1})\dots E(i_{\frac{r}{2}},j_{\frac{r}{2}})$, $ A_{l}^{'}=\begin{pmatrix} \alpha_{1}\\ \vdots \\ \alpha_{r} \end{pmatrix}$, where $E(i,j)$ is an elementary commutative matrix and $\alpha_{i}~(i\in\{1,\dots,r\})$ is the row vector of $A_{l}^{'}$, then\\
$$(\eta_{1},\dots,\eta_{r})=\eta=\mathbf{a}A_{l}^{'}=(a_{1},\dots,a_{r})\begin{pmatrix} \alpha_{1}\\ \vdots \\ \alpha_{r} \end{pmatrix}=a_{1}\alpha_{1}+\dots +a_{r}\alpha_{r}.$$

For any $t\in\{1,2,\dots,r\}$, we have

\begin{equation*}
\begin{split}
\eta_{t}&=a_{1}\alpha_{1t}+\dots +a_{r}\alpha_{rt}\\
&=a_{i_{1}}\alpha_{i_{1},t}+a_{j_{1}}\alpha_{j_{1},t}+\dots+a_{i_{\frac{r}{2}}}\alpha_{i_{\frac{r}{2}},t}+a_{j_{\frac{r}{2}}}\alpha_{j_{\frac{r}{2}},t}.
\end{split}
\end{equation*}

Let
$$(\beta_{1},\dots,\beta_{r})=\beta=\mathbf{a}PA_{l}^{'}=a_{i_{1}}\alpha_{j_{1}}+a_{j_{1}}\alpha_{i_{1}}+\dots+
a_{i_{\frac{r}{2}}}\alpha_{j_{\frac{r}{2}}}+a_{j_{\frac{r}{2}}}\alpha_{i_{\frac{r}{2}}},$$
thus
$$\beta_{t}=a_{i_{1}}\alpha_{j_{1},t}+a_{j_{1}}\alpha_{i_{1},t}+\dots +a_{i_{\frac{r}{2}}}\alpha_{j_{\frac{r}{2}},t}+a_{j_{\frac{r}{2}}}\alpha_{i_{\frac{r}{2}},t},$$
therefore
$$\beta_{r-t+1}=a_{i_{1}}\alpha_{j_{1},r-t+1}+a_{j_{1}}\alpha_{i_{1},r-t+1}+\dots +a_{i_{\frac{r}{2}}}\alpha_{j_{\frac{r}{2}},r-t+1}+a_{j_{\frac{r}{2}}}\alpha_{i_{\frac{r}{2}},r-t+1}.$$

Since the reverse of the $i_{k}$th row is the $j_{k}$th row ($k \in \{1,2,\dots,\frac{r}{2}\}$), we have $\alpha_{i_{k},t}=\alpha_{j_{k},r-t+1}$, then
\begin{equation*}
\beta_{r-t+1}=a_{i_{1}}\alpha_{i_{1},t}+a_{j_{1}}\alpha_{j_{1},t}+\dots +a_{i_{\frac{r}{2}}}\alpha_{i_{\frac{r}{2}},t}+a_{j_{\frac{r}{2}}}\alpha_{j_{\frac{r}{2}},t}=\eta_{t},
\end{equation*}
i.e.:  $\beta=\eta^{r}$. Therefore we have $(\mathbf{a}A_{l}^{'})^{r}=\mathbf{a}PA_{l}^{'}$. \par
Since the elements of the group $H_{i}$ are listed as form Equation \eqref{eq.3.1}, let $r=2s$, we have
$$
\{h_{i1},h_{i2}, \dots,h_{ir}\}=\{e,t_{1},\dots,t_{s-1},\beta t_{s-1},\dots,\beta t_{1},\beta \},
$$
then
$$
\{h_{i1}^{-1},h_{i2}^{-1},\dots,h_{ir}^{-1}\}=\{e,t_{1}^{-1},\dots,t_{s-1}^{-1},(\beta t_{s-1})^{-1},\dots,(\beta t_{1})^{-1},\beta\}.
$$

We can easily calculate that in $A_{l}^{'}$ the reverse of the first row is the last row. Since $\phi_{i,l}$ is a bijection and from the structure of $A_{l}^{'}$, we know that we only need to prove that before mapping $\phi_{i,l}$, for any $u \in\{2,\dots,r-1\}$, the reverse of the elements $\{h_{iu}^{-1}h_{iv}\}$ in $u$ row is the elements of a particular row, where $ v \in\{1,\dots,r\}$. For $1 \leq m \leq s-1$, the elements $\{h_{i,m+1}^{-1}h_{iv}\}$ in $(m+1)$-th row are as follows:
$$t_{m}^{-1},~t_{m}^{-1}t_{1},\dots,~t_{m}^{-1}t_{s-1},~t_{m}^{-1}\beta t_{s-1},\dots,~t_{m}^{-1}\beta t_{1},~t_{m}^{-1}\beta.$$
The elements $\{h_{i,r-m}^{-1}h_{iv}\}$ in $(r-m)$-th row are:
$$(\beta t_{m})^{-1},(\beta t_{m})^{-1}t_{1},\dots,(\beta t_{m})^{-1}t_{s-1},(\beta t_{m})^{-1}\beta t_{s-1},\dots, (\beta t_{m})^{-1}\beta t_{1},(\beta t_{m})^{-1}\beta .$$
Since $\beta^{-1}=\beta$, so we can simplify it to the following:
$$t_{m}^{-1}\beta,~t_{m}^{-1}\beta t_{1},\dots,~t_{m}^{-1}\beta t_{s-1},~t_{m}^{-1}t_{s-1},\dots,~t_{m}^{-1}t_{1},t_{m}^{-1},$$
which is the reverse of the elements in $(m+1)$-th row. Therefore we can obtain that the reverse of the $m$th row is the $(r-m+1)$th row, where $1\leq m \leq l$. Then for all finite group $H_{i}$ with the form as Equation \eqref{eq.3.1}, the permutation matrix $P$ are the same.
\end{proof}

Now we prove Theorem \ref{th.3.3}:

\begin{proof}[Proof of Theorem~\ref{th.3.3}]
Let $\mathcal{D}^{\ast}(v) =\langle \Omega^{\ast}(v)\rangle $ be a linear composite $G$-code as defined in \eqref{eq.3.3}. Since the partitioned matrix $\sigma(v)$ is reversible block matrix, then the matrix $\Omega^{\ast}(v)$ that obtained using our method is still a reversible block matrix. For each code ${\bf c} \in \mathcal{D}^{\ast}(v)$, there is a partitioned vector $(y_{1},y_{2},\dots,y_{\frac{n}{r}}) \in R^{n}$, where $y_{i}\in R^{r}, i\in\{1,\dots,\frac{n}{r}\}$, s.t.,
\begin{gather*}
\begin{split}
{\bf c}&=(y_{1},y_{2},\dots,y_{\frac{n}{r}})\begin{pmatrix}
A_{1}^{'}&A_{2}^{'}&\dots  &A_{\frac{n}{r}}^{'}\\
A_{\frac{n}{r}+1}^{'} &A_{\frac{n}{r}+2}^{'}&\dots&A_{\frac{2n}{r}}^{'}\\
\vdots   &\vdots  & &\vdots\\
A_{(\frac{n}{r}-1)\frac{n}{r}+1}^{'} &A_{(\frac{n}{r}-1)\frac{n}{r}+2}^{'}&\dots  &A_{\frac{n^{2}}{r^{2}}}^{'}
\end{pmatrix} \\
&=(\sum_{i=1}^{\frac{n}{r}}y_{i}A_{(i-1)\frac{n}{r}+1}^{'},~\sum_{i=1}^{\frac{n}{r}}y_{i}A_{(i-1)\frac{n}{r}+2}^{'},~\dots~,
~\sum_{i=1}^{\frac{n}{r}}y_{i}A_{\frac{in}{r}}^{'}).
\end{split}
\end{gather*}

Thus
\begin{gather*}
\begin{split}
{\bf c}^{r}&=(\sum_{i=1}^{\frac{n}{r}}y_{i}A_{(i-1)\frac{n}{r}+1}^{'},~\sum_{i=1}^{\frac{n}{r}}y_{i}A_{(i-1)\frac{n}{r}+2}^{'},~\dots~,
~\sum_{i=1}^{\frac{n}{r}}y_{i}A_{\frac{in}{r}}^{'})^{r}\\
&=((\sum_{i=1}^{\frac{n}{r}}y_{i}A_{\frac{in}{r}}^{'})^{r},~(\sum_{i=1}^{\frac{n}{r}}y_{\frac{n}{r}}A_{\frac{in}{r}-1}^{'})^{r},~\dots~,
~(\sum_{i=1}^{\frac{n}{r}}y_{i}A_{(i-1)\frac{n}{r}+1}^{'})^{r}).
\end{split}
\end{gather*}

By Lemma \ref{lemma.3.2}, we have $(y_{i}A_{l}^{'})^{r}=y_{i}PA_{l}^{'}$. Assume that in the block matrix $\Omega^{\ast}(v)$, the reverse of the $i$th row is the $i^{'}$th row, where $i,i^{'} \in \{1,2,\dots,\frac{n}{r}\}$ and $i\neq i^{'}$. Therefore

\begin{gather*}
\begin{split}
{\bf c}^{r}&=(\sum_{i=1}^{\frac{n}{r}}y_{i}PA_{\frac{in}{r}}^{'},~\sum_{i=1}^{\frac{n}{r}}y_{i}PA_{\frac{in}{r}-1}^{'},~\dots~,
~\sum_{i=1}^{\frac{n}{r}}y_{i}PA_{(i-1)\frac{n}{r}+1}^{'})\\
&=(y_{1^{'}}P,y_{2^{'}}P,\dots,y_{(\frac{n}{r})^{'}}P)\begin{pmatrix}
A_{1}^{'}&A_{2}^{'}&\dots  &A_{\frac{n}{r}}^{'}\\
A_{\frac{n}{r}+1}^{'} &A_{\frac{n}{r}+2}^{'}&\dots&A_{\frac{2n}{r}}^{'}\\
\vdots   &\vdots& &\vdots\\
A_{(\frac{n}{r}-1)\frac{n}{r}+1}^{'} &A_{(\frac{n}{r}-1)\frac{n}{r}+2}^{'}&\dots  &A_{\frac{n^{2}}{r^{2}}}^{'}
\end{pmatrix}.
\end{split}
\end{gather*}
So ${\bf c}^{r} \in \mathcal{D}^{\ast}(v)$, i.e., the code $\mathcal{D}^{\ast}(v)$ in $R^{n}$ is a reversible code of index 1.
\end{proof}

\section{Reversible Composite Matrices over Ring $R$}
In this section, we employ the group $G$ of order $n$ from Theorem \ref{th.3.2}, and the construction method in Section $3$. For the cases of $\frac{n}{r}=2$ and $\frac{n}{r}=4$, we use three special forms of groups $H_{i}, i\in\{1,2,3\}$ to construct some reversible composite matrices in a finite ring $R$, which we then use to generate reversible DNA codes.\par
\subsection{The Case of $\frac{n}{r}=2$}
Let $G = \langle x, y~|~ x^{\frac{n}{2}} = y^2 = 1, xy = yx \rangle= C_{2} \times C_{\frac{n}{2}}$,  where $n$ is a positive integer that is divisible by $16$, let
\begin{center}
$v=\sum\limits_{i=0}^{\frac{n}{2}-1} [ \alpha_{g_{i+1}}x^{i}+\alpha_{g_{i+(\frac{n}{2}+1)}}x^{i}y] \in R(C_{2}\times C_{\frac{n}{2}}).$
\end{center}
By Theorem~\ref{th.3.2}, the partitioned matrix
$$\sigma(v)=\begin{pmatrix}	
A&B\\
B&A
\end{pmatrix}
$$
is a reversible block matrix, where $A = circ(\alpha_{g_{1}},\dots,\alpha_{g_{\frac{n}{2}}})$ and $B =circ(\alpha_{g_{\frac{n}{2}+1}},\dots,\alpha_{g_{n}})$.

We will now present the three forms of composite matrices that we obtained.\par

{\bf (I). The reversible composite matrix generated from $H_1$ and $H_2$}\par
Let $H_{1} =\langle a,b~|~a^{\frac{n}{4}}=b^{2}=1,a^{b}=a^{-1} \rangle$ be the dihedral group of order $\frac{n}{2}$ with its elements being listed as follows (in accordance to Equation \eqref{eq.3.1}):
\begin{center}
$\{ h_{11},h_{12},\dots,h_{1\frac{n}{2}}\} =\{ 1,a,a^2,\dots,a^{\frac{n}{4}-1},ba^{\frac{n}{4}-1},ba^{\frac{n}{4}-2},\dots,ba,b \},$
\end{center}
and let $H_{2} =\langle c,d~|~c^{\frac{n}{4}}=d^{2}=1,c^{d}=c^{\frac{n}{8}-1} \rangle$ be the quasi-dihedral group of order $\frac{n}{2}$ with its elements being listed as follows (in accordance to Equation \eqref{eq.3.1}:
\begin{center}
$\{ h_{21},h_{22},\dots,h_{2\frac{n}{2}}\}
=\{ 1,c,c^2,\dots,c^{\frac{n}{4}-1},dc^{\frac{n}{4}-1},dc^{\frac{n}{4}-2},\dots,dc,d \}.$
\end{center}

Next, using the method given in Section $3$, we obtain the matrix $\Omega^{\ast}(v)$ in the following form:
$$
\sigma(v)=\begin{pmatrix} 	
A&B\\
B&A
\end{pmatrix}
\longrightarrow
\Omega^{\ast}(v)=\begin{pmatrix}
A_{1}^{'}&A_{2}^{'}\\
 A_{3}^{'}&A_{4}^{'}
 \end{pmatrix},
$$
where
$$
A_{1}^{'}=\begin{pmatrix}
\alpha_{g_{1}^{-1}g_{1}}&\alpha_{g_{1}^{-1}g_{2}}&\dots&\alpha_{g_{1}^{-1}g_{\frac{n}{2}}}\\
\alpha_{\phi_{1,1}(h_{12}^{-1}h_{11})}&\alpha_{\phi_{1,1}(h_{12}^{-1}h_{12})}
&\dots&\alpha_{\phi_{1,1}(h_{12}^{-1}h_{1\frac{n}{2}})}\\
\vdots   &\vdots           &      &\vdots\\
\alpha_{\phi_{1,1}(h_{1\frac{n}{2}}^{-1}h_{11})}&\alpha_{\phi_{1,1}(h_{1\frac{n}{2}}^{-1}h_{12})}
&\dots &\alpha_{\phi_{1,1}(h^{-1}_{1\frac{n}{2}}h_{1\frac{n}{2}})}
\end{pmatrix}
$$
with $\phi_{1,1}:h_{1i} \mapsto g_{1}^{-1}g_{i}$ when  $ i = 1,2,\dots,\frac{n}{2},$
$$
A_{2}^{'}=\begin{pmatrix}
\alpha_{g_{1}^{-1}g_{\frac{n}{2}+1}}&\alpha_{g_{1}^{-1}g_{\frac{n}{2}+2}}&\dots&\alpha_{g_{1}^{-1}g_{n}}\\
\alpha_{\phi_{2,2}(h_{22}^{-1}h_{21})}&\alpha_{\phi_{2,2}(h_{22}^{-1}h_{22})}
&\dots&\alpha_{\phi_{2,2}(h_{22}^{-1}h_{2\frac{n}{2}})}\\
\vdots   &\vdots            &      &\vdots\\
\alpha_{\phi_{2,2}(h_{2\frac{n}{2}}^{-1}h_{21})}&\alpha_{\phi_{2,2}(h_{2\frac{n}{2}}^{-1}h_{22})}
&\dots&\alpha_{\phi_{2,2}(h_{2\frac{n}{2}}^{-1}h_{2\frac{n}{2}})}
\end{pmatrix}
$$
with $\phi_{2,2}:h_{2i} \mapsto g_{1}^{-1}g_{\frac{n}{2}+i} $ when $  i = 1,2,\dots,\frac{n}{2},$
$$
A_{3}^{'}=\begin{pmatrix}
\alpha_{g_{\frac{n}{2}+1}^{-1}g_{1}}&\alpha_{g_{\frac{n}{2}+1}^{-1}g_{2}}&\dots&\alpha_{g_{\frac{n}{2}+1}^{-1}g_{\frac{n}{2}}}\\
\alpha_{\phi_{2,3}(h_{22}^{-1}h_{21})}&\alpha_{\phi_{2,3}(h_{22}^{-1}h_{22})}
&\dots&\alpha_{\phi_{2,3}(h_{22}^{-1}h_{2\frac{n}{2}})}\\
\vdots   &\vdots         &      &\vdots\\
\alpha_{\phi_{2,3}(h_{2\frac{n}{2}}^{-1}h_{21})}&\alpha_{\phi_{2,3}(h_{2\frac{n}{2}}^{-1}h_{22})}
&\dots&\alpha_{\phi_{2,3}(h_{2\frac{n}{2}}^{-1}h_{2\frac{n}{2}})}\\
\end{pmatrix}
$$
with $\phi_{2,3}:h_{2i} \mapsto g_{\frac{n}{2}+1}^{-1}g_{i}$ when  $i = 1,2,\dots,\frac{n}{2},$ and

$$
A_{4}^{'}=\begin{pmatrix}
\alpha_{g_{\frac{n}{2}+1}^{-1}g_{\frac{n}{2}+1}}&\alpha_{g_{\frac{n}{2}+1}^{-1}g_{\frac{n}{2}+2}}
&\dots&\alpha_{g_{\frac{n}{2}+1}^{-1}g_{n}}\\
\alpha_{\phi_{1,4}(h_{12}^{-1}h_{11})}&\alpha_{\phi_{1,4}(h_{12}^{-1}h_{12})}
&\dots&\alpha_{\phi_{1,4}(h_{12}^{-1}h_{1\frac{n}{2}})}\\
\vdots   &\vdots           &      &\vdots\\
\alpha_{\phi_{1,4}(h_{1\frac{n}{2}}^{-1}h_{11})}&\alpha_{\phi_{1,4}(h_{1\frac{n}{2}}^{-1}h_{12})}
&\dots&\alpha_{\phi_{1,4}(h_{1\frac{n}{2}}^{-1}h_{1\frac{n}{2}})}\\
\end{pmatrix}
$$
with $\phi_{1,4}:h_{1i} \mapsto g_{\frac{n}{2}+1}^{-1}g_{\frac{n}{2}+i}  $
when $i = 1,2,\dots,\frac{n}{2}.$   We have\par

\begin{Proposition}
Assume the notation is as given above. Using the groups $H_1$ and $H_2$, we obtain a reversible composite matrix with the following form:
$$\mathcal{G}_{12}=\Omega^{\ast}(v)
=\begin{pmatrix} 	
A_{1}&B_{1}&A_{2}&B_{2}\\
B_{1}^{T}&A_{1}^{T}&B_{2}^{F}&A_{2}^{T}\\
A_{2}&B_{2}&A_{1}&B_{1}\\
B_{2}^{F}&A_{2}^{T}&B_{1}^{T}&A_{1}^{T}\\
\end{pmatrix},
$$
where
\begin{align*}
&A_{1} = circ(\alpha_{g_{1}}, \alpha_{g_{2}} ,\dots,\alpha_{g_{\frac{n}{4}}}),\\
&B_{1} = circ(\alpha_{g_{\frac{n}{4}+1}}, \alpha_{g_{\frac{n}{4}+2}} ,\dots,\alpha_{g_{\frac{n}{2}}}),\\
&A_{2} = circ(\alpha_{g_{\frac{n}{2}+1}}, \alpha_{g_{\frac{n}{2}+2}} ,\dots,\alpha_{g_{\frac{3n}{4}}}), \\
&B_{2} = (\frac{n}{8}+1)\text{-}circ(\alpha_{g_{\frac{3n}{4}+1}}, \alpha_{g_{\frac{3n}{4}+2}} ,\dots,\alpha_{g_{n}}).
\end{align*}

\end{Proposition}

{\bf (II). The reversible composite matrix generated from $H_2$ and $H_3$}\par
Let $H_{2} =\langle c,d~|~ c^{\frac{n}{4}}=d^{2}=1,c^{d}=c^{\frac{n}{8}-1} \rangle$ be the quasi-dihedral group of order $\frac{n}{2}$ with its elements being listed as follows (in accordance to Equation \eqref{eq.3.1}):
\begin{center}
$\{h_{21},h_{22},\dots,h_{2\frac{n}{2}}\}
=\{1,c,c^2,\dots,c^{\frac{n}{4}-1},dc^{\frac{n}{4}-1},dc^{\frac{n}{4}-2},\dots,dc,d \},$
\end{center}
and let $H_{3} =\langle e~|~e^{\frac{n}{2}}=1\rangle$ be the cyclic group of order $\frac{n}{2}$ with its elements being listed as follows (in accordance to Equation \eqref{eq.3.1}):
\begin{center}
$\{h_{31},h_{32},\dots,h_{3\frac{n}{2}}\} =\{1,e^2,e^4,\dots,e^{\frac{n}{2}-2},e^{\frac{n}{4}}e^{\frac{n}{2}-2},e^{\frac{n}{4}}e^{\frac{n}{2}-4},
\dots,e^{\frac{n}{4}}e^{2},e^{\frac{n}{4}} \}.$
\end{center}

Next, using the method given in Section $3$, we can obtain the matrix $\Omega^{\ast}(v)$ in the following form:
$$
\sigma(v)=\begin{pmatrix} 	
A&B\\
B&A
\end{pmatrix}
\longrightarrow
\Omega^{\ast}(v)=\begin{pmatrix}
A_{1}^{'}&A_{2}^{'}\\
 A_{3}^{'}&A_{4}^{'}
 \end{pmatrix},
$$
where
$$
A_{1}^{'}=\begin{pmatrix}
\alpha_{g_{1}^{-1}g_{1}}&\alpha_{g_{1}^{-1}g_{2}}&\dots&\alpha_{g_{1}^{-1}g_{\frac{n}{2}}}\\
\alpha_{\phi_{3,1}(h_{32}^{-1}h_{31})}&\alpha_{\phi_{3,1}(h_{32}^{-1}h_{32})}
&\dots&\alpha_{\phi_{3,1}(h_{32}^{-1}h_{3\frac{n}{2}})}\\
\vdots   &\vdots         &      &\vdots\\
\alpha_{\phi_{3,1}(h_{3\frac{n}{2}}^{-1}h_{31})}&\alpha_{\phi_{3,1}(h_{3\frac{n}{2}}^{-1}h_{32})}
&\dots&\alpha_{\phi_{3,1}(h_{3\frac{n}{2}}^{-1}h_{3\frac{n}{2}})}\\
\end{pmatrix}
$$
with $\phi_{3,1}:h_{3i} \mapsto g_{1}^{-1}g_{i}$ when  $ i = 1,2,\dots,\frac{n}{2},$
$$
A_{2}^{'}=\begin{pmatrix}
\alpha_{g_{1}^{-1}g_{\frac{n}{2}+1}}&\alpha_{g_{1}^{-1}g_{\frac{n}{2}+2}}&\dots&\alpha_{g_{1}^{-1}g_{n}}\\
\alpha_{\phi_{2,2}(h_{22}^{-1}h_{21})}&\alpha_{\phi_{2,2}(h_{22}^{-1}h_{22})}
&\dots&\alpha_{\phi_{2,2}(h_{22}^{-1}h_{2\frac{n}{2}})}\\
\vdots   &\vdots        &      &\vdots\\
\alpha_{\phi_{2,2}(h_{2\frac{n}{2}}^{-1}h_{21})}&\alpha_{\phi_{2,2}(h_{2\frac{n}{2}}^{-1}h_{22})}
&\dots&\alpha_{\phi_{2,2}(h_{2\frac{n}{2}}^{-1}h_{2\frac{n}{2}})}\\
\end{pmatrix}
$$
with $\phi_{2,2}:h_{2i} \mapsto g_{1}^{-1}g_{\frac{n}{2}+i} $ when $  i = 1,2,\dots,\frac{n}{2},$

$$
A_{3}^{'}=\begin{pmatrix}
\alpha_{g_{\frac{n}{2}+1}^{-1}g_{1}}&\alpha_{g_{\frac{n}{2}+1}^{-1}g_{2}}&\dots&\alpha_{g_{\frac{n}{2}+1}^{-1}g_{\frac{n}{2}}}\\
\alpha_{\phi_{2,3}(h_{22}^{-1}h_{21})}&\alpha_{\phi_{2,3}(h_{22}^{-1}h_{22})}
&\dots&\alpha_{\phi_{2,3}(h_{22}^{-1}h_{2\frac{n}{2}})}\\
\vdots   &\vdots           &      &\vdots\\
\alpha_{\phi_{2,3}(h_{2\frac{n}{2}}^{-1}h_{21})}&\alpha_{\phi_{2,3}(h_{2\frac{n}{2}}^{-1}h_{22})}
&\dots &\alpha_{\phi_{2,3}(h_{2\frac{n}{2}}^{-1}h_{2\frac{n}{2}})}\\
\end{pmatrix}
$$
with $\phi_{2,3}:h_{2i} \mapsto g_{\frac{n}{2}+1}^{-1}g_{i}$ when  $i = 1,2,\dots,\frac{n}{2},$ and
$$
A_{4}^{'}=\begin{pmatrix}
\alpha_{g_{\frac{n}{2}+1}^{-1}g_{\frac{n}{2}+1}}&\alpha_{g_{\frac{n}{2}+1}^{-1}g_{\frac{n}{2}+2}}
&\dots&\alpha_{g_{\frac{n}{2}+1}^{-1}g_{n}}\\
\alpha_{\phi_{3,4}(h_{32}^{-1}h_{31})}&\alpha_{\phi_{3,4}(h_{32}^{-1}h_{32})}
&\dots&\alpha_{\phi_{3,4}(h_{32}^{-1}h_{3\frac{n}{2}})}\\
\vdots   &\vdots          &      &\vdots\\
\alpha_{\phi_{3,4}(h_{3\frac{n}{2}}^{-1}h_{31})}&\alpha_{\phi_{3,4}(h_{3\frac{n}{2}}^{-1}h_{32})}
&\dots &\alpha_{\phi_{3,4}(h_{3\frac{n}{2}}^{-1}h_{3\frac{n}{2}})}\\
\end{pmatrix}
$$
with $\phi_{3,4}:h_{3i} \mapsto g_{\frac{n}{2}+1}^{-1}g_{\frac{n}{2}+i} $ when $i = 1,2,\dots,\frac{n}{2}.$  We have\par

\begin{Proposition}
Assume the notation is as given above. Using groups $H_2$ and $H_3$, we obtain a reversible composite matrix with the following form:
$$\mathcal{G}_{32}=\Omega^{\ast}(v)
=\begin{pmatrix} 	
A_{1}&B_{1}&A_{3}&B_{3}\\
B_{2}&A_{2}&B_{3}^{F}&A_{3}^{T}\\
A_{3}&B_{3}&A_{1}&B_{1}\\
B_{3}^{F}&A_{3}^{T}&B_{2}&A_{2}\\
\end{pmatrix},
$$
where
\begin{align*}
&A_{1} = circ(\alpha_{g_{1}}, \alpha_{g_{2}} ,\dots,\alpha_{g_{\frac{n}{4}}}),\\
&B_{1} = revcirc(\alpha_{g_{\frac{n}{4}+1}}, \alpha_{g_{\frac{n}{4}+2}} ,\dots,\alpha_{g_{\frac{n}{2}}}),\\
&A_{2}= circ(\alpha_{g_{1}},\alpha_{g_{\frac{n}{4}}},\alpha_{g_{\frac{n}{4}-1}},\dots,\alpha_{g_{2}} ),\\
&B_{2} = revcirc(\alpha_{g_{\frac{n}{2}-1}}, \alpha_{g_{\frac{n}{2}-2}} ,\dots, \alpha_{g_{\frac{n}{4}+1}},\alpha_{g_{\frac{n}{2}}}),\\
&A_{3} = circ(\alpha_{g_{\frac{n}{2}+1}}, \alpha_{g_{\frac{n}{2}+2}} ,\dots,\alpha_{g_{\frac{3n}{4}}}), \\
&B_{3} = (\frac{n}{8}+1)\text{-}circ(\alpha_{g_{\frac{3n}{4}+1}}, \alpha_{g_{\frac{3n}{4}+2}} ,\dots,\alpha_{g_{n}}).
\end{align*}

\end{Proposition}

{\bf (III). The reversible composite matrix generated from $H_2$}\par
Let $H_{2} =\langle c,d~|~ c^{\frac{n}{4}}=d^{2}=1,c^{d}=c^{\frac{n}{8}-1} \rangle$ be the quasi-dihedral group of order $\frac{n}{2}$ with its elements being listed as follows (in accordance to Equation \eqref{eq.3.1}):
\begin{center}
$\{ h_{21},h_{22},\dots,h_{2\frac{n}{2}}\}
=\{ 1,c,c^2,\dots,c^{\frac{n}{4}-1},dc^{\frac{n}{4}-1},dc^{\frac{n}{4}-2},\dots,dc,d \}.$
\end{center}

Next, using the method given in Section $3$, we obtain the matrix $\Omega^{\ast}(v)$ in the following form:
$$
\sigma(v)=\begin{pmatrix} 	
A&B\\
B&A
\end{pmatrix}
\longrightarrow
\Omega^{\ast}(v)=\begin{pmatrix}
A_{1}^{'}&A_{2}^{'}\\
 A_{3}^{'}&A_{4}^{'}
 \end{pmatrix},
$$
where
$$
A_{1}^{'}=\begin{pmatrix}
\alpha_{g_{1}^{-1}g_{1}}&\alpha_{g_{1}^{-1}g_{2}}&\dots&\alpha_{g_{1}^{-1}g_{\frac{n}{2}}}\\
\alpha_{\phi_{2,1}(h_{22}^{-1}h_{21})}&\alpha_{\phi_{2,1}(h_{22}^{-1}h_{22})}
&\dots&\alpha_{\phi_{2,1}(h_{22}^{-1}h_{2\frac{n}{2}})}\\
\vdots   &\vdots            &      &\vdots\\
\alpha_{\phi_{2,1}(h_{2\frac{n}{2}}^{-1}h_{21})}&\alpha_{\phi_{2,1}(h_{2\frac{n}{2}}^{-1}h_{22})}
&\dots &\alpha_{\phi_{2,1}(h_{2\frac{n}{2}}^{-1}h_{2\frac{n}{2}})}\\
\end{pmatrix}
$$
with $\phi_{2,1}:h_{2i} \mapsto g_{1}^{-1}g_{i}$ when  $ i = 1,2,\dots,\frac{n}{2},$
$$
A_{2}^{'}=\begin{pmatrix}
\alpha_{g_{1}^{-1}g_{\frac{n}{2}+1}}&\alpha_{g_{1}^{-1}g_{\frac{n}{2}+2}}&\dots&\alpha_{g_{1}^{-1}g_{n}}\\
\alpha_{\phi_{2,2}(h_{22}^{-1}h_{21})}&\alpha_{\phi_{2,2}(h_{22}^{-1}h_{22})}
&\dots&\alpha_{\phi_{2,2}(h_{22}^{-1}h_{2\frac{n}{2}})}\\
\vdots   &\vdots            &      &\vdots\\
\alpha_{\phi_{2,2}(h_{2\frac{n}{2}}^{-1}h_{21})}&\alpha_{\phi_{2,2}(h_{2\frac{n}{2}}^{-1}h_{22})}
&\dots&\alpha_{\phi_{2,2}(h_{2\frac{n}{2}}^{-1}h_{2\frac{n}{2}})}\\
\end{pmatrix}
$$
with $\phi_{2,2}:h_{2i} \mapsto g_{1}^{-1}g_{\frac{n}{2}+i} $ when $  i = 1,2,\dots,\frac{n}{2},$

$$
A_{3}^{'}=\begin{pmatrix}
\alpha_{g_{\frac{n}{2}+1}^{-1}g_{1}}&\alpha_{g_{\frac{n}{2}+1}^{-1}g_{2}}&\dots&\alpha_{g_{\frac{n}{2}+1}^{-1}g_{\frac{n}{2}}}\\
\alpha_{\phi_{2,3}(h_{22}^{-1}h_{21})}&\alpha_{\phi_{2,3}(h_{22}^{-1}h_{22})}
&\dots&\alpha_{\phi_{2,3}(h_{22}^{-1}h_{2\frac{n}{2}})}\\
\vdots   &\vdots           &      &\vdots\\
\alpha_{\phi_{2,3}(h_{2\frac{n}{2}}^{-1}h_{21})}&\alpha_{\phi_{2,3}(h_{2\frac{n}{2}}^{-1}h_{22})}
&\dots&\alpha_{\phi_{2,3}(h_{2\frac{n}{2}}^{-1}h_{2\frac{n}{2}})}\\
\end{pmatrix}
$$
with $\phi_{2,3}:h_{2i} \mapsto g_{\frac{n}{2}+1}^{-1}g_{i}$ when  $i = 1,2,\dots,\frac{n}{2},$ and

$$
A_{4}^{'}=\begin{pmatrix}
\alpha_{g_{\frac{n}{2}+1}^{-1}g_{\frac{n}{2}+1}}&\alpha_{g_{\frac{n}{2}+1}^{-1}g_{\frac{n}{2}+2}}
&\dots&\alpha_{g_{\frac{n}{2}+1}^{-1}g_{n}}\\
\alpha_{\phi_{2,4}(h_{22}^{-1}h_{21})}&\alpha_{\phi_{2,4}(h_{22}^{-1}h_{22})}
&\dots&\alpha_{\phi_{2,4}(h_{22}^{-1}h_{2\frac{n}{2}})}\\
\vdots   &\vdots        &      &\vdots\\
\alpha_{\phi_{2,4}(h_{2\frac{n}{2}}^{-1}h_{21})}&\alpha_{\phi_{2,4}(h_{2\frac{n}{2}}^{-1}h_{22})}
&\dots &\alpha_{\phi_{2,4}(h_{2\frac{n}{2}}^{-1}h_{2\frac{n}{2}})}\\
\end{pmatrix}
$$
with $\phi_{2,4}:h_{2i} \mapsto g_{\frac{n}{2}+1}^{-1}g_{\frac{n}{2}+i}$ when $i = 1,2,\dots,\frac{n}{2}.$ We have\par

\begin{Proposition}
Assume the notation is as given above. Using the group $H_2$, we obtain a reversible composite matrix with the following form:
$$\mathcal{G}_{22}=\Omega^{\ast}(v)
=\begin{pmatrix} 	
A_{1}&B_{1}&A_{2}&B_{2}\\
B_{1}^{F}&A_{1}^{T}&B_{2}^{F}&A_{2}^{T}\\
A_{2}&B_{2}&A_{1}&B_{1}\\
B_{2}^{F}&A_{2}^{T}&B_{1}^{F}&A_{1}^{T}\\
\end{pmatrix},
$$
where
\begin{align*}
&A_{1} = circ(\alpha_{g_{1}}, \alpha_{g_{2}} ,\dots,\alpha_{g_{\frac{n}{4}}}),\\
&B_{1} = (\frac{n}{8}+1)\text{-}circ(\alpha_{g_{\frac{n}{4}+1}}, \alpha_{g_{\frac{n}{4}+2}},\dots,\alpha_{g_{\frac{n}{2}}}),\\
&A_{2} = circ(\alpha_{g_{\frac{n}{2}+1}}, \alpha_{g_{\frac{n}{2}+2}} ,\dots,\alpha_{g_{\frac{3n}{4}}}), \\
&B_{2} = (\frac{n}{8}+1)\text{-}circ(\alpha_{g_{\frac{3n}{4}+1}}, \alpha_{g_{\frac{3n}{4}+2}} ,\dots,\alpha_{g_{n}}).
\end{align*}

\end{Proposition}

\subsection{The Case of $\frac{n}{r}=4$}
Let $G = \langle x, y~|~ x^{\frac{n}{4}} = y^4 = 1, xy = yx \rangle= C_{4} \times C_{\frac{n}{4}}$, where $n$ is a positive integer that is divisible by $16$, and let
\begin{center}
$v_{1}=\sum\limits_{i=0}^{\frac{n}{4}-1}[ \alpha_{g_{i+1}}x^{i}+\alpha_{g_{i+(\frac{3n}{4}+1)}}x^{i}y^{2}
 +\alpha_{g_{i+\frac{n}{4}+1}}x^{i}y+\alpha_{g_{i+\frac{n}{2}+1}}x^{i}y^{3}] \in R(C_{4}\times C_{\frac{n}{4}}).$
\end{center}

By Theorem~\ref{th.3.2}, the partitioned matrix
$$\sigma(v_{1})=\begin{pmatrix}	
A&B&C&D\\
C&A&D&B\\
B&D&A&C\\
D&C&B&A
\end{pmatrix}
$$
is a reversible block matrix, where $A=circ(\alpha_{g_{1}},\dots,\alpha_{g_{\frac{n}{4}}})$, $B=circ(\alpha_{g_{\frac{n}{4}+1}},\dots,$ $\alpha_{g_{\frac{n}{2}}})$, $C=circ(\alpha_{g_{\frac{n}{2}+1}},\dots,\alpha_{g_{\frac{3n}{4}}})$,
and $D=circ(\alpha_{g_{\frac{3n}{4}+1}},\dots,\alpha_{g_{n}})$. \par
In this subsection, for the case of $\frac{n}{r}=4$, we also use the three groups $H_{i}$ from the Subsection $4.1$ and the method given in Section $3$ to generate some reversible composite matrices and later obtain some DNA codes using these matrices. Since there are many different types of composite matrices, we will not describe the precise form of each composite matrix. We use $G_{ijkl}$ to denote the composite matrix obtained by constructing $A'_l$ in the four blocks of the first row of $\sigma(v)$ using $H_{i}$, $H_{j}$, $H_{k}$ and $H_{l}$ ($i,j,k,l \in \{1, 2, 3\}$), respectively.
For example, $G_{1123}$ is a composite matrix obtained by constructing $A'_l$ in the four blocks of the first row of $\sigma(v)$ using $H_{1}$, $H_{1}$, $H_{2}$ and $H_{3}$ respectively.

The following example is to  illustrate the above construction.\par

\begin{Example}
Assume $G= \langle x,y~|~x^4=y^4=1,xy=yx\rangle $, $R=\mathbb{F}_4$ and $v\in RG$ in the form of Equation\eqref{eq.3.2}. Let $\mathbf{G}_{1111}= \Omega^{\ast}(v)$ be a composite matrix generated from $H_1$ using the method from Subsection $4.2$ and its form is as follows:
$$
\addtocounter{MaxMatrixCols}{10}
\mathbf{G}_{1111}=\left(
\begin{array}{cc:cc:cc:cc}
A_{1}&B_{1}&A_{2}&B_{2}&A_{3}&B_{3}&A_{4}&B_{4}\\
B_{1}^{T}&A_{1}^{T}&B_{2}^{T}&A_{2}^{T}&B_{3}^{T}&A_{3}^{T}&B_{4}^{T}&A_{4}^{T}\\
\hdashline
A_{3}&B_{3}&A_{1}&B_{1}&A_{4}&B_{4}&A_{2}&B_{2}\\
B_{3}^{T}&A_{3}^{T}&B_{1}^{T}&A_{1}^{T}&B_{4}^{T}&A_{4}^{T}&B_{2}^{T}&A_{2}^{T}\\
\hdashline
A_{2}&B_{2}&A_{4}&B_{4}&A_{1}&B_{1}&A_{3}&B_{3}\\
B_{2}^{T}&A_{2}^{T}&B_{4}^{T}&A_{4}^{T}&B_{1}^{T}&A_{1}^{T}&B_{3}^{T}&A_{3}^{T}\\
\hdashline
A_{4}&B_{4}&A_{3}&B_{3}&A_{2}&B_{2}&A_{1}&B_{1}\\
B_{4}^{T}&A_{4}^{T}&B_{3}^{T}&A_{3}^{T}&B_{2}^{T}&A_{2}^{T}&B_{1}^{T}&A_{1}^{T}\\
\end{array}
\right)
$$
$$
=
\addtocounter{MaxMatrixCols}{10}
\left(
\begin{array}{cccc:cccc:cccc:cccc}
1&1&\omega^{2}&0&0&\omega^{2}&\omega^{2} &\omega^{2}&\omega &\omega&\omega^2 &0&\omega &\omega &1&\omega \\
1&1&0&\omega^2&   \omega^2&0&\omega^2&\omega^2&   \omega&\omega&0&\omega^2&   \omega&\omega&\omega&1\\
\omega^2&0&1&1&   \omega^2&\omega^2&0&\omega^2&   \omega^2&0&\omega&\omega&   1&\omega&\omega&\omega\\
0&\omega^2&1&1&   \omega^2&\omega^2&\omega^2&0&   0&\omega^2&\omega&\omega&   \omega&1&\omega&\omega\\
\hdashline
\omega&\omega&\omega^2&0&    1&1&\omega^2&0&   \omega&\omega&1&\omega&   0&\omega^2&\omega^2&\omega^2\\
\omega&\omega&0&\omega^2&    1&1&0&\omega^2&   \omega&\omega&\omega&1&   \omega^2&0&\omega^2&\omega^2\\
\omega^2&0&\omega&\omega&    \omega^2&0&1&1&   1&\omega&\omega&\omega&   \omega^2&\omega^2&0&\omega^2\\
0&\omega^2&\omega&\omega&    0&\omega^2&1&1&   \omega&1&\omega&\omega&   \omega^2&\omega^2&\omega^2&0\\
\hdashline
0&\omega^2&\omega^2&\omega^2&   \omega&\omega&1&\omega&    1&1&\omega^2&0&   w&w&w^2&0\\
\omega^2&0&\omega^2&\omega^2&   \omega&\omega&\omega&1&    1&1&0&\omega^2&   w&w&0&w^2\\
\omega^2&\omega^2&0&\omega^2&   1&\omega&\omega&\omega&    \omega^2&0&1&1&   w^2&0&w&w\\
\omega^2&\omega^2&\omega^2&0&   \omega&1&\omega&\omega&    0&\omega^2&1&1&   0&w^2&w&w\\
\hdashline
\omega&\omega&1&\omega&   \omega&\omega&\omega^2&0&   0&\omega^2&\omega^2&\omega^2&   1&1&\omega^2&0\\
\omega&\omega&\omega&1&   \omega&\omega&0&\omega^2&   \omega^2&0&\omega^2&\omega^2&   1&1&0&\omega^2\\
1&\omega&\omega&\omega&   \omega^2&0&\omega&\omega&   \omega^2&\omega^2&0&\omega^2&   \omega^2&0&1&1\\
\omega&1&\omega&\omega&   0&\omega^2&\omega&\omega&   \omega^2&\omega^2&\omega^2&0&   0&\omega^2&1&1\\
\end{array}
\right),
$$
where $A_{1} = circ(1, 1), B_{1} = circ(\omega^2, 0), A_{2} = circ(0, \omega^2), B_{2} = circ(\omega^2, \omega^2), A_{3} = circ(\omega, \omega)$, and $B_{3} = circ(\omega^2,0), A_{4} = circ(\omega, \omega), B_{4} = circ(1, \omega)$. Let $\mathcal{D}^{\ast}=\langle \mathbf{G}_{1111} \rangle$, we can observe that $\mathbf{1}=(1,1,\dots,1) \in \mathcal{D}^{\ast}$. Then by Magma, we know that there are $65536$ codewords in $\mathcal{D}^{\ast}$ satisfying the HD, RV, and RC constraints with $d = 6$. We can also calculate the $GC$-weight enumerator of $\mathcal{D}^{\ast}$ using Magma:
\begin{center}
$GCW_{\mathcal{D}^{\ast}}(X_{1}, X_{2}) =16X_{1}^{16} + 128X_{1}^{14}X_{2}^2 + 4032X_{1}^{12}X_{2}^4 + 15232X_{1}^{10}X_{2}^6 + 26720X_{1}^8X_{2}^8
+ 15232X_{1}^6X_{2}^{10} + 4032X_{1}^4X_{2}^{12} + 128X_{1}^2X_{2}^{14}+ 16X_{2}^{16}$.
\end{center}

Let $D=\eta(\mathcal{D}^{\ast})$ be the  DNA code which obtained by employing the map $\eta$ to the codewords of $\mathcal{D}^{\ast}$. By the GC-weight enumerator, we know that $D$ has $26720$ codewords that satisfies the HD, RV, RC, and GC constraints. We can employ Algorithm $1$ in~\cite{b19} to the DNA code $D$ to filter out $8$-conflict free DNA codes with length of $16$, in which each DNA codewords satisfies the HD, RV, RC, and GC constraints ($d =6$) and in which each DNA string is free from secondary structure. The detailed calculation methods, algorithms and results, please see~\cite{b19}. We will not elaborate here.\par
\end{Example}

\section{Computational Results}
In this section, we use the matrices $\mathcal{G}_{ij}$ and $\mathbf{G}_{ijkl}$ (where $i,j,k,l \in \{1,2,3\}$) obtained from Subsections $4.1$ and $4.2$ to search for DNA codes of different lengths. We perform our searches in the software package MAGMA~\cite{b8} and Matlab. Since the length $n$ is large, for each DNA code $D$, we only compare the maximum size of $D$ of length $n$ satisfying the HD and RC constraints for a given $d$. Many of our lower bounds are better than the current best known bounds. We use asterisks to represent data that is equal to the current best known bound, diamond to represent data that is better than the best known bound, and dagger to represent data that has not available in previous literature. Since the DNA codes constructed using our method satisfy reversibility, if we want the codewords of $D$ satisfying RC constraint, we only need to ensure that the vector $\mathbf{1}=(1,1,\dots,1)\in D$. Since there are no results on  the lower bounds on $A_{4}^{RC}(n, d)$ for $n$ = $80$, $96$, $160$ in the literature, we obtain some new results for these three lengths of DNA codes. Additionally, the lower bounds of most of the codes in Tables $1$ and $2$ are improved for fixed parameters $n$ and $d$.\par

\section{Conclusion}

In this paper, we present a new method for constructing reversible DNA codes. We provide a specific form of group $G$, which ensures that our construction result is a general construction. We employ our construction method to generate some DNA codes, and compare the lower bounds of these DNA codes that satisfy certain constraints to previous data. Many of our codes have better lower bounds on the sizes of some known DNA codes. We also construct, new to the literature, DNA codes of lengths $80$, $96$, and $160$ for some fixed Hamming distance $d$ that satisfy some constraints.

\vskip 3mm
{\bf Conflict of interest}: The authors declare that there is no conflict of interest.

\vskip 3mm
{\bf Acknowledgement.}

This work was supported by National Natural Science Foundation of China under Grant Nos.12271199 and 12171191 and The Fundamental Research Funds for the Central Universities 30106220482. Whan-Hyuk Choi is supported by the National Research Foundation of Korea under Grant NRF-2022R1C1C2011689.

\begin{sidewaystable}[htp]
\caption{Lower bounds on $A_{4}^{RC}(n, d)$ from $\mathcal{G}_{ij}$, where $i,j \in \{ 1,2,3\}$}
\centering
\begin{threeparttable}
\tiny
\begin{tabular}{ccccc}
 \toprule
 \text{Generator Matrix} & $n$ & $d$ & $A_{4}^{RC}(n,d)$ & \text{Best Known}$A_{4}^{RC}(n,d)$ \\
 \midrule
 $\mathcal{G}_{12}$ & 32 & 4&$281474976710656^{\ast} $&281474976710656\cite{b19}\\
 $\mathcal{G}_{12}$ & 48 & 3&$ 1208925819614629174706176^{\dagger}$ &-\\
 $\mathcal{G}_{12}$ & 48 & 4&$75557863725914323419136^{\ast}$ &75557863725914323419136\cite{b19}\\
 $\mathcal{G}_{12}$ & 64 & 4&$20282409603651670423947251286016^{\ast} $&20282409603651670423947251286016\cite{b19}\\
 $\mathcal{G}_{12}$ & 80 & 4&$ 87112285931760246646623899502532662132736^{\dagger}$ &-\\
 $\mathcal{G}_{12}$ & 96 & 2&$ 24519928653854221733733552434404946937899825954937634816^{\dagger} $&-\\
 $\mathcal{G}_{12}$ & 112 & 2&$105312291668557186697918027683670432318895095400549111254310977536^{\diamond}$ &6582018229284824168619876730229402019930943462534319453394436096\cite{b19}\\
 $\mathcal{G}_{12}$ & 128 & 2&$452312848583266388373324160190187140051835877600158453279131187530910662656^{\ast}$ &452312848583266388373324160190187140051835877600158453279131187530910662656\cite{b19}\\
 \hline
 $\mathcal{G}_{22}$ & 32 & 4&$281474976710656^{\ast}$ &281474976710656\cite{b19}\\
 $\mathcal{G}_{22}$ & 48 & 7&$ 1152921504606846976^{\dagger} $&-  \\
 $\mathcal{G}_{22}$ & 64 & 4&$20282409603651670423947251286016^{\ast}$&20282409603651670423947251286016\cite{b19}\\
 $\mathcal{G}_{22}$ & 80 & 3&$ 22300745198530623141535718272648361505980416^{\dagger}$&- \\
 $\mathcal{G}_{22}$ & 80 & 4&$ 340282366920938463463374607431768211456^{\dagger}$ &-\\
 $\mathcal{G}_{22}$ & 96 & 2&$ 24519928653854221733733552434404946937899825954937634816^{\dagger}$ &-\\
 $\mathcal{G}_{22}$ & 112 & 2&$105312291668557186697918027683670432318895095400549111254310977536^{\diamond}$ &6582018229284824168619876730229402019930943462534319453394436096\cite{b19}\\
 $\mathcal{G}_{22}$ & 128 & 2&$452312848583266388373324160190187140051835877600158453279131187530910662656^{\ast}$ &452312848583266388373324160190187140051835877600158453279131187530910662656\cite{b19}\\
 \hline
 $\mathcal{G}_{32}$ & 32 & 4&$281474976710656^{\ast}$ &281474976710656\cite{b19}\\
 $\mathcal{G}_{32}$ & 48 & 3&$ 1208925819614629174706176^{\dagger}$ &-\\
 $\mathcal{G}_{32}$ & 80 & 3&$ 22300745198530623141535718272648361505980416^{\dagger}$ &-\\
 $\mathcal{G}_{32}$ & 80 & 4&$ 87112285931760246646623899502532662132736^{\dagger}$ &-\\
 $\mathcal{G}_{32}$ & 96 & 2&$ 24519928653854221733733552434404946937899825954937634816^{\dagger} $&-\\
 $\mathcal{G}_{32}$ & 96 & 4&$ 23384026197294446691258957323460528314494920687616^{\dagger}$ &-\\
 $\mathcal{G}_{32}$ & 112 & 2&$105312291668557186697918027683670432318895095400549111254310977536^{\diamond} $ &6582018229284824168619876730229402019930943462534319453394436096\cite{b19}\\
\bottomrule
\end{tabular}

\begin{tablenotes}
\footnotesize
\item[]$\ast$: equal to the currently best known data; $\diamond$: better than currently best known data; $\dagger$: regarding the new $d$'s data.
\end{tablenotes}
\end{threeparttable}
\end{sidewaystable}

\begin{sidewaystable}[htp]
\caption{Lower bounds on $A_{4}^{RC}(n, d)$ from $\mathbf{G}_{ijkl}$, where $i,j,k,l \in \{ 1,2,3\}$}
\centering
\begin{threeparttable}
\tiny
\begin{tabular}{ccccc}
 \toprule
 \text{Generator Matrix} & $n$ & $d$ & $A_{4}^{RC}(n,d)$ & \text{Best Known}$A_{4}^{RC}(n,d)$ \\
 \midrule
$\mathbf{G}_{2222}$ & 32 & 3&$ 281474976710656^{\dagger}$ &-\\
$\mathbf{G}_{1111}$ & 32 & 4&$281474976710656^{\ast}$ &281474976710656\cite{b19}\\
$\mathbf{G}_{1111}$ & 48 & 3&$ 1208925819614629174706176^{\dagger}$ &-\\

$\mathbf{G}_{1111}$ & 80 & 4&$ 87112285931760246646623899502532662132736^{\dagger}$ &-\\
$\mathbf{G}_{2222}$ & 96 & 2&$ 95780971304118053647396689196894323976171195136475136^{\dagger}$ &-\\
$\mathbf{G}_{2222}$ & 96 & 3&$ 23384026197294446691258957323460528314494920687616^{\dagger}$ &-\\
$\mathbf{G}_{1111}$ & 112 & 2&$105312291668557186697918027683670432318895095400549111254310977536^{\diamond}$ &6582018229284824168619876730229402019930943462534319453394436096\cite{b19}\\
$\mathbf{G}_{1111}$ & 128 & 4&$ 26959946667150639794667015087019630673637144422540572481103610249216^{\dagger}$ &-\\
$\mathbf{G}_{1111}$ & 160 & 3&$ 32592575621351777380295131014550050576823494298654980010178247189670100796213387298934358016^{\dagger}$ &-\\
\hline
$\mathbf{G}_{2111}$ & 32 & 4&$281474976710656^{\ast}$ &281474976710656\cite{b19}\\
$\mathbf{G}_{3111}$ & 48 & 3&$1208925819614629174706176^{\dagger}$ &-\\
$\mathbf{G}_{1121}$ & 64 & 4&$1267650600228229401496703205376 $ &20282409603651670423947251286016\cite{b19}\\
$\mathbf{G}_{3111}$ & 80 & 4&$ 87112285931760246646623899502532662132736^{\dagger}$ &-\\
$\mathbf{G}_{1222}$ & 96 & 3&$95780971304118053647396689196894323976171195136475136^{\dagger}$ &-\\
$\mathbf{G}_{3111}$ & 112 & 2&$105312291668557186697918027683670432318895095400549111254310977536
^{\diamond}$ &6582018229284824168619876730229402019930943462534319453394436096\cite{b19}\\
$\mathbf{G}_{1121}$ & 128 & 4&$26959946667150639794667015087019630673637144422540572481103610249216^{\dagger}$&-\\
$\mathbf{G}_{3111}$ & 160 & 2&$32592575621351777380295131014550050576823494298654980010178247189670100796213387298934358016^{\dagger}$&-\\
$\mathbf{G}_{2111}$ & 160 & 3&$32592575621351777380295131014550050576823494298654980010178247189670100796213387298934358016^{\dagger}$&-\\
\hline
$\mathbf{G}_{3132}$ & 32 & 4&$281474976710656^{\ast}$ &281474976710656\cite{b19}\\
$\mathbf{G}_{2113}$ & 64 & 4&$1267650600228229401496703205376 $ &20282409603651670423947251286016\cite{b19}\\
$\mathbf{G}_{1322}$ & 96 & 2&$ 95780971304118053647396689196894323976171195136475136^{\dagger}$ &-\\
$\mathbf{G}_{3121}$ & 128 & 4&$26959946667150639794667015087019630673637144422540572481103610249216^{\dagger}$&-\\
$\mathbf{G}_{3121}$ & 160 & 2&$32592575621351777380295131014550050576823494298654980010178247189670100796213387298934358016^{\dagger}$
&-\\
\bottomrule
\end{tabular}

\begin{tablenotes}
\footnotesize
\item[]$\ast$: equal to the currently best known data; $\diamond$: better than currently best known data; $\dagger$: regarding the new $d$'s data; no markings: no better than preceding data.
\end{tablenotes}
\end{threeparttable}
\end{sidewaystable}

\vskip 4mm

\newpage

\end{document}